\documentclass[pre,longbibliography,twocolumn,aps,showpacs,showkeys]{revtex4-1}

\usepackage{amssymb, amsmath, amsfonts}
\usepackage{multirow}
\usepackage{graphicx}
\usepackage{dcolumn}
\usepackage{bm}
\usepackage{float}
\usepackage[caption=false]{subfig} 
\usepackage{xcolor}
\usepackage{hyperref} 
\hypersetup{bookmarksnumbered= true, bookmarks= true, pdfstartview= FitH,%
            breaklinks= true, colorlinks= true, citecolor= blue, linkcolor= blue, urlcolor= blue}
\usepackage{breakurl}

\let\revappendix\appendix
\usepackage[capitalise]{cleveref}

\crefformat{equation}{Eq.~(#2#1#3)}
\crefname{enumi}{}{}
\crefformat{app_sec}{the Appendix} 

\crefrangeformat{equation}{Eqs.~(#3#1#4)--(#5#2#6)}
\crefrangeformat{definition}{Definitions~#3#1#4--#5#2#6}
\crefrangeformat{lemma}{Lemmas~#3#1#4--#5#2#6}
\crefrangeformat{theorem}{Theorems~#3#1#4--#5#2#6}
\crefrangeformat{corollary}{Corollaries~#3#1#4--#5#2#6}
\crefrangeformat{remark}{Remarks~#3#1#4--#5#2#6}
\crefrangeformat{conjecture}{Conjectures~#3#1#4--#5#2#6}
\crefrangeformat{enumi}{#3#1#4--#5#2#6}

\usepackage{amsthm}
\theoremstyle{definition} 
\newtheorem{theorem}{Theorem}[section]            

\newtheorem{proposition}{Proposition}[section]    

\usepackage[inline]{enumitem}

\newcommand{\talpha}{\tilde{\alpha}}
\newcommand{\balpha}{\bar{\alpha}}
\newcommand{\sgn}{\text{sgn}}

\begin{document}
\title{Diversity of hysteresis in a fully cooperative coinfection model}

\author{Jorge P. Rodr\'iguez}
    \affiliation{Instituto de F\'isica Interdisciplinar y Sistemas Complejos IFISC, CSIC-UIB, Palma de Mallorca, Spain}	
	
\author{Yu-Hao Liang}
\author{Jonq Juang}
    \email{jjuang@math.nctu.edu.tw}
    \affiliation{Department of Applied Mathematics, National Chiao Tung University, Hsinchu, Taiwan}

\begin{abstract}
    We propose a fully cooperative coinfection model in which singly infected individuals are more likely to acquire a second disease than those who are susceptible, and doubly infected individuals are also assumed to be more contagious than those infected with one disease. The dynamics of such fully cooperative coinfection model between two interacting infectious diseases is investigated through well-mixed and network-based approaches. We show that the former approach exhibits three types of hysteresis, namely, $C$, $S_l$ and $S_r$ types, where the last two types have not been identified before. The first (resp., the second and the third) type exhibits (resp., exhibit) discontinuous outbreak transition from the disease free (resp., low prevalence) state to the high prevalence state when a transmission rate crosses a threshold from the below. Moreover, the third (resp., the first and the second) type possesses (resp., possess) discontinuous eradication transition from the high prevalence state to the low prevalence (resp., disease free) state when the transmission rate reaches a threshold from the above. Complete characterization of these three types of hysteresis in term of parameters measuring the uniformity of the model is also provided. Finally, we assess numerically this epidemic dynamics in random networks.
\end{abstract}



\maketitle

\date{\today}

\section{INTRODUCTION}

  Mathematical modelling has become a fundamental tool to understand the basic features of disease spreading in a group \cite{Hethcote2000}. We can deal with this problem with two approaches: considering a well-mixed system and approximating it with an interaction all-to-all, {\it i.e.,} the mean-field approximation, or taking into account a topology that specifies who interacts with whom. Hence, the theory of both dynamical systems and complex networks \cite{Newman2003} is fundamental for this kind of studies. Two main models have been proposed for the dynamics of a spreading disease: one, considering three kinds of individuals, where susceptibles (S), when exposed to the disease, become infective (I), and afterwards become recovered (R), getting immunized against the disease, is known as SIR model \cite{KermacK1927}; the other model represents diseases that are spread by pathogens that frequently mutate, such that individuals do not get immunized, but come back to susceptible state after the infective phase (SIS model) \cite{Anderson1992Book}. In fact, those models have been proposed not only in the context of disease spreading, but also for other spreading processes, such as computer viruses epidemics \cite{Pastor-Satorras2001A}.

  Most studies on disease spreading focus on one disease infecting a population \cite{Pastor-Satorras2001B,Moreno2002,Newman2002}. However, different pathogens may interact with each other on the host. For instance, there are cases in which infection with one disease strengthens a host's immunity against other infection \cite{Newman2005,Ahn2006,Funk2010C,Marceau2011,Karrer2011,Poletto2013,Sahneh2014}. In those competitive systems, regimes such as extinction, coexistence and one strain dominance have been identified. On the other hand, it has also been reported that people with HIV infections are more likely to get other infections such as hepatitis or tuberculosis \cite{Alter2006,Abu-Raddad2006,Pawlowski2012,WHO2014}. Such cooperative disease spreading models have recently been studied in Refs.\ \cite{Hethcote2000,Pastor-Satorras2001A,Newman2005,Ahn2006,Alter2006,Abu-Raddad2006,Colizza2006,Barrat2008Book,Funk2010C,Marceau2011,Karrer2011,Pawlowski2012,Poletto2013,Chen2013,WHO2014,Sahneh2014,Sanz2014,Cai2015,Grassberger2016,Chen2016}. Interesting dynamics, like the appearance of abrupt transitions, in contrast to the typical second order phase transition found in single disease spreading model, have been reported \cite{Cai2015,Grassberger2016,Chen2016}. Specifically, for SIS dynamics, there is hysteresis behaviour: the models exhibit discontinuous outbreak transition from disease free state to the high prevalence state as a transmission rate crosses a threshold $\alpha_o$ from the below, while in the reverse scenario, they exhibit discontinuous eradication transition from the high prevalence state to disease free state as the transmission rate reaches a threshold $\alpha_e\, (< \alpha_o)$ from the above is also observed. The above described phase transition is to be called a \emph{hysteresis of $C$ type}. At $\alpha= \alpha_o$, we see that the size of the infected population has a discontinuous jump, indicating the outbreak of the disease. The difference in $\alpha_e$ and $\alpha_o$ points out that, once the outbreak of the disease occurs, driving down the transmission rate to $\alpha_o$ is not enough to eradicate the diseases. In fact, we need to drive the rate further down to $\alpha_e$ for the epidemic to die out, which requires more effort and results in greater economic costs. Note that, in the cooperative coinfection models considered in Refs.\ \cite{Chen2013,Cai2015,Grassberger2016,Chen2016}, it was only assumed that a singly infected individual is more likely to get infected with a second disease than a susceptible individual.

  In this work, we consider a fully cooperative coinfection model for two diseases, in which we further assume that a doubly infected individual is more contagious than a singly infected individual; see \cref{fig:General_Scheme}. Hence, diseases cooperate not only when one individual is getting infected, but also when they meet in one host and start infecting other. We observe two additional types of hysteresis. In particular, the model has discontinuous outbreak transition from the low prevalence state to the high prevalence state as the transmission rate crosses a threshold $\alpha_o$ from the below. Moreover, a discontinuous eradication transition from the high prevalence state to the disease free state or the lower prevalence state can both be observed as the transmission rate reaches a threshold $\alpha_e\, (< \alpha_o)$ from the above. They are to be called \emph{hysteresis of $S_l$ and $S_r$ types}, respectively; see \cref{fig:Hysteresis_Loops}. It is worthwhile to point out that, for model exhibiting hysteresis of $S_l$ or $S_r$ types, it must exist two stable endemic fixed points in some range of transmission rates as contrast to the existence of at most one stable endemic fixed point for all range of transmission rates in the case of hysteresis of $C$ type.

  To depict the above mentioned dynamics, we first consider the model through the well-mixed approach. Three positive quantities $C$, $a$, $b$, called the relative infection factor, relative coinfection factors for susceptible class and singly infected class, respectively, are to be introduced to measure the uniformity of the model; see \cref{fig:General_Scheme}. When $C=a=b=1$, the model is said to be uniform and corresponds to noninteracting disease spreading model. The larger deviation of these quantities from $1$, the more nonuniform the model is. The case that $\min \{ C,a,b \} >1$ (resp., $\max \{ C,a,b \} <1$) corresponds to a fully cooperative (resp., competitive) coinfection model. If $\min \{ C,a,b \} < 1< \max \{ C,a,b \}$, it is called a mixed type. In Ref.\ \cite{Chen2016}, the authors consider a partially cooperative model, \textit{i.e.}, $C>1$ and $a,b=1$, where the hysteresis of type $C$ was found provided that $C$ is sufficiently large. In this work, for sufficiently large $C$ and any $a,b\geq 1$, the model behaves like the one considered in Ref.\ \cite{Chen2016}. In case that $b-a> 1/2$ and $C$ in an intermediate range, hysteresis of types $S_l$ and $S_r$ emerge. Subsequently, we investigate this problem through the network-based approach.

\begin{figure}[htp]
    \subfloat[\label{fig:General_Scheme}]{ \includegraphics[scale=1]{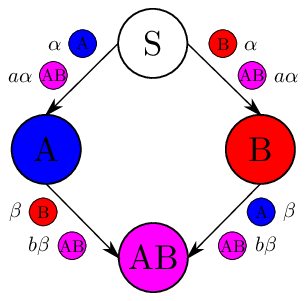} }  
             \hspace*{6pt}
    \subfloat[\label{fig:Hysteresis_Loops}]{ \includegraphics[scale=1]{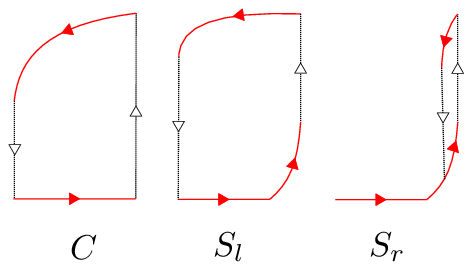} } 
    \caption{ (a) General scheme of the dynamics where, under the exposure to infected individuals, the transmission rates from one state to others are specified. (b) Hysteresis loops found under the considered dynamics.}
    \label{fig:scheme}
\end{figure}

The organization of the paper is as follows. In \cref{sec:Well-Mixed MODEL}, the well-mixed approach for our fully cooperative coinfection model is introduced and analytically studied. In \cref{sec:STOCHASTIC_NETWORK_MODEL}, we describe our stochastic network model and the numerical simulation results are also provided. The summary of our obtained results and the future work are given in \cref{sec:CONCLUSION}. All proofs of our results are recorded in \cref{sec:Appendix_Existence_of_bar_alpha}.

\section{Well-Mixed approach}\label{sec:Well-Mixed MODEL}

We study a well-mixed particle system for the fully cooperative coinfection dynamics with two diseases, A and B. We define $S$, $I_{A}$, $I_{B}$ and $I_{AB}$ as, respectively, the fractions of susceptible individuals and the ones infected with disease A, B and AB in the population. Susceptible individuals get a primary infection from a singly infected with a rate $\alpha$. The cooperation between diseases on individuals that are singly infected is translated into a higher rate for getting secondary infection, $\beta > \alpha$. Moreover, the cooperation makes the doubly infected individuals more virulent, such that they infect susceptibles with a rate $a \alpha$, $a \geq 1$ and primary infected with a rate $ b \beta$, $b \geq 1$. The case $a=b=1$ has been previously reported \cite{Chen2016}, implying cooperation for getting a secondary infection. Our case is more general, translating cooperation between the diseases not only for receiving, but also for transmitting diseases. Then the corresponding mean field equations for this dynamical system are:

\begin{equation} \label{eq:Well_Mixed_Model}
    \begin{aligned}
        \dot{S}(t)        &= -\alpha   S (I_A+ a I_{AB})- \alpha  S (I_B+ a I_{AB})+  I_A   + I_B,\\
        \dot{I}_{A}(t)    &=  \alpha   S (I_A+ a I_{AB})- \beta I_A (I_B+ b I_{AB})+  I_{AB}- I_A,\\
        \dot{I}_{B}(t)    &=  \alpha   S (I_B+ a I_{AB})- \beta I_B (I_A+ b I_{AB})+  I_{AB}- I_B,\\
        \dot{I}_{AB}(t)   &=  \beta  I_A (I_B+ b I_{AB})+ \beta I_B (I_A+ b I_{AB})- 2I_{AB}.
    \end{aligned}
\end{equation}
Here the recovery rates for diseases A and B, without loss of generality, are set equal to the time scale, while the basic reaction scheme underlying the dynamical process is given in the first and second columns of \cref{table:Relation_p_alpha}; see also \cref{fig:General_Scheme}. More precisely, $\alpha$ and $\alpha a$ (resp., $\beta$ and $\beta b$) are infection rates for susceptible class (resp., singly infected class) who get infected by singly infected and doubly infected
class (resp., the complementarily infected class and doubly infected class), respectively. The quantity $a$ (resp., $b$) is the ratio between the infection rates of the doubly infected class and singly infected class for infecting the susceptible class (resp., singly infected class). The larger $a$ and $b$ are, the more contagious the doubly infected class AB is. We call $a$ and $b$ the \emph{relative coinfection factors} for susceptible class and singly infected class, respectively. The quantity $C := \beta / \alpha$, which is termed the \emph{relative infection factor}, is the ratio of the infection rates for singly infected class and susceptible class. The larger $C$ is, the easier singly infected class gets infected with a secondary infection. Note that $C$, $a$ and $b$ measure the degree of uniformity in fully cooperative coinfection models. The larger those quantities deviate from $1$, the more nonumiform the model is.

\begin{table}[!hbt]
\begin{tabular}{|l||c|c|}
  \hline
  \multirow{2}{*}{Transmission} & \multirow{2}{78pt}{\hspace{18pt}Well-mixed \hspace*{28pt}  (rate)}%
   & \multirow{2}{78pt}{Stochastic network \hspace*{12pt} (probability) \hspace{4pt}} \\[1pt]
   &  & \\
  \hline
  $S \xrightarrow{A}  A$     & $\alpha$                & $p$\\
  $S \xrightarrow{B}  B$     & $\alpha$                & $p$ \\
  $S \xrightarrow{AB} A$     & $a\alpha$               & $1-(1-p)^a$ \\
  $S \xrightarrow{AB} B$     & $a\alpha$               & $1-(1-p)^a$ \\[3pt]
  $A \xrightarrow{B}  AB$    & $\beta\, (=C\alpha)$    & $1-(1-p)^C$ \\
  $B \xrightarrow{A}  AB$    & $\beta\, (=C\alpha)$    &  $1-(1-p)^C$ \\
  $A \xrightarrow{AB} AB$    & $b\beta\, (=bC\alpha)$  & $1-(1-p)^{bC}$\\[3pt]
  $B \xrightarrow{AB} AB$    & $b\beta\, (=bC\alpha)$  & $1-(1-p)^{bC}$  \\
  \hline
\end{tabular}
\caption{Basic reaction scheme for the disease spreading in the well-mixed model and stochastic network models.
         \label{table:Relation_p_alpha}}
\end{table}

The dynamics of this model is depicted by its fixed points. Assuming a symmetric behaviour between the diseases A and B, the endemic fixed point $E^*:= (S^*, I_A^*, I_B^*, I_{AB}^*)$ of \eqref{eq:Well_Mixed_Model} satisfies
\begin{equation}\label{eq:endemic_fixed_pt_Is}
    \begin{aligned}
        I_A^*     = I_B^* &= \frac{a \alpha S^*}{1+(2a-1)\alpha S^*} (1-S^*),\\
        I_{AB}^* &= \frac{1-\alpha S^*}{1+(2a-1)\alpha S^*} (1-S^*),
    \end{aligned}
\end{equation}
where $S^*$, satisfying $0\leq S^* < 1$ and $0\leq \alpha S^* \leq 1$, is a solution of
\begin{equation}\label{eq:endemic_fixed_pt_S}
    \begin{aligned}
        f(S, \alpha) := &a(b-a) C \alpha^3 S^3\\
                 &+[ \alpha a(a-b)C - abC+ 2a- 1 ] \alpha^2 S^2\\
                 &+(\alpha abC -2a+ 2) \alpha S -1 =0.
    \end{aligned}
\end{equation}

For any $a,b\geq 1$ and $C>0$, Eq.\ \eqref{eq:endemic_fixed_pt_S} defines a solution curve $(S, \balpha(S))$, $0< S\leq 1$, satisfying $f(S, \balpha(S))= 0$. In particular, for any fixed $a,b\geq 1$, $C>0$ and $0< S\leq 1$, there exists a unique $\balpha\, (= \balpha(S))$ on $(0, 1/S]$ such that $f(S, \balpha(S))= 0$ (see \cref{prop:exist_func_balpha} in \cref{sec:Appendix_Existence_of_bar_alpha}). Our goal next is to study the graph the function $\balpha(S)$, $0< S\leq 1$. We address it by making the change of variable
\begin{equation}\label{eq:defi_T}
    T= \alpha S.
\end{equation}
Clearly, by \eqref{eq:endemic_fixed_pt_S},
\begin{equation}\label{eq:defi_tilde_alpha}
    \begin{aligned}
        \bar{\alpha}(S) &= \frac{a(a-b)C T^3+ (abC-2a+1) T^2+ 2(a-1)T+ 1}{aCT [(a-b)T+ b]}\\
                        &=: \talpha(T).
    \end{aligned}
\end{equation}
Here $0< T \leq 1$. Using \cref{prop:relation_balpha_talpha} in \cref{sec:Appendix_Existence_of_bar_alpha}, we have that the monotonicity structure of functions $\balpha(S)$ and $\talpha(T)$ on the interval $(0,1]$ is identical. Hence, to study the dependence of the endemic fixed points on the parameters, it suffices to graph the function of $\talpha(T)$. The detailed analysis of the graph of $\talpha(T)$ is provided in \cref{prop:diff_talpha_endpt,prop:b_less_than_b_c,prop:defi_C1_C_12,prop:talpha_property_various_C} in \cref{sec:Appendix_Existence_of_bar_alpha}. We summarize the results in the following theorem.

\begin{theorem}\label{thm:main_results}
    For parameters $a$, $b$, $C\geq 1$, $\lim\limits_{T\rightarrow 0^+} \talpha(T)= \infty$ and $\talpha(1)= 1$. Define $C_2\, (=C_2(a,b))$ by
    \begin{equation}\label{eq:defi_C_2}
        C_2= \frac{2}{a}.
    \end{equation}
    Then the following assertions \cref{itm:b_less_than_b_c,itm:b_greater_than_b_c} hold.
    \begin{enumerate}[label=(\Roman*),ref=(\Roman*)]
        \item\label{itm:b_less_than_b_c}
            (See \cref{fig:TAlpha_Func_b_1p25,fig:EvoBifCurve_b_1p25})
            For $b\leq  b_c\, (:=a+ 1/2)$, $\talpha(T)$ is concave upward on $(0,1]$. Moreover, the following two assertions hold.
            \begin{enumerate}[label=(\Roman{enumi}-\roman*),ref=(\Roman{enumi}-\roman*)]
                \item For $C\leq C_2$, $\talpha(T)$ is strictly decreasing on $(0,1]$. Consequently, the endemic fixed point curve, denoted by $\tilde{S}$, on the $(\alpha, 1-S)$ plane is strictly increasing.
                \item For $C>    C_2$, $\talpha(T)$ has exactly one critical point at some point $T_*$ on $(0,1)$, which corresponds to a local minimum $\alpha_e\, (=\alpha_e(C))$. Moreover, $\alpha_e< \talpha(1) \,(=1)$. Consequently, the corresponding $\tilde{S}$, on the $(\alpha, 1-S)$ plane is $C$-shaped.
            \end{enumerate}
        \item\label{itm:b_greater_than_b_c}
             (See \cref{fig:TAlpha_Func_b_2,fig:EvoBifCurve_b_2}) For $b> b_c$, $\talpha(T)$ has exactly one inflection point $\gamma$ on $(0,1)$. Define $C_1\, (=C_1(a,b))$ and $C_{12}\, (=C_{12}(a,b))$, respectively, by
             \begin{align}
                \mbox{ }\hskip 25pt
                C_1    &= \frac{ \left[ a^{2/3} (2b-1)^{1/3} + (b-a)^{2/3} \right]^3}{ab^3}, \label{eq:defi_C_1}\\
                C_{12} &= \frac{ (2ab-2a+b)+ 2 \sqrt{a (2b-1) (b-a)}}{ab^2}. \label{eq:defi_C_12}
             \end{align}
             Then the following three assertions hold.
             \begin{enumerate}[label=(\Roman{enumi}-\roman*),ref=(\Roman{enumi}-\roman*)]
                \item For $C\leq C_1$, $\talpha(T)$ is strictly decreasing on $(0,1]$. Consequently, the corresponding $\tilde{S}$, on the $(\alpha, 1-S)$ plane is strictly increasing.
                \item For $C_1 < C < C_{2}$, $\talpha(T)$ has a unique local minimum $\alpha_e\, (=\alpha_e(C))$ and maximum $\alpha_o\, (=\alpha_o(C))$ at $T_*$ and $T^*$, respectively, on $(0,1)$ with $T_* < T^*$ such that $\talpha(T)$ is strictly decreasing on $(0, T_*] \cup [T^*, 1]$ while strictly increasing on $(T_*, T^*)$. Moreover,
                        \begin{equation}\label{eq:C_12_relation}
                            \alpha_e (C)
                                \left\{
                                \begin{array}{cl}
                                    > 1 & \text{if } C< C_{12}, \\[1pt]
                                    = 1 & \text{if } C= C_{12}, \\[1pt]
                                    < 1 & \text{if } C> C_{12}. \\[1pt]
                                \end{array}
                                \right.
                        \end{equation}
                      Consequently, the corresponding $\tilde{S}$, on the $(\alpha, 1-S)$ plane is $S$-shaped.
                \item For $C> C_2$, $\talpha(T)$ has exactly one critical point at some point $T_*$ on $(0,1)$, which corresponds to a local minimum $\alpha_e\, (=\alpha_e(C))$. Moreover, $\alpha_e< \talpha(1) \,(=1)$. Consequently, the corresponding $\tilde{S}$, on the $(\alpha, 1-S)$ plane is $C$-shaped.
            \end{enumerate}
    \end{enumerate}
\end{theorem}

\begin{figure}[htp]
    \subfloat[ $(a,b)= (1, 1.25)$ and $C$ varies discretely from $1.5< 1.75< C_2(1, 5/4)\, (=2) < 2.25< 2.5$, which correspond to the five curves in the panel from the right to the left, respectively.  \label{fig:TAlpha_Func_b_1p25}]%
             { \includegraphics[scale=1]{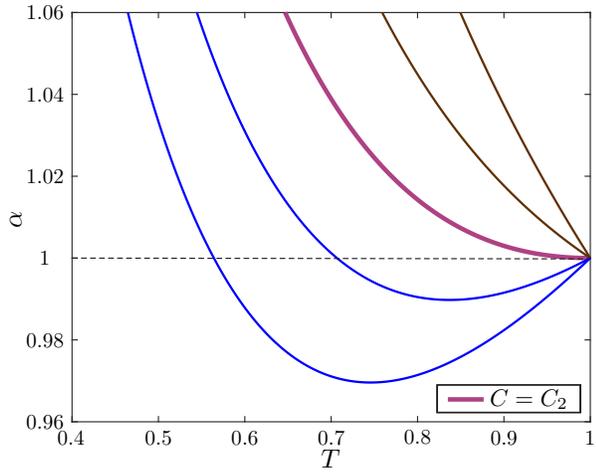} } 
    \\
    \subfloat[ $(a,b)= (1, 2)$ and $C$ varies discretely from $1.75< C_1(1,2)\, (\approx 1.821) < 1.85 < C_{12}(1,2)\, (\approx 1.866) < 1.92 < C_2(1,2)\, (=2)< 2.03$, which correspond to the seven curves in the panel from the right to the left, respectively.\label{fig:TAlpha_Func_b_2}]%
             { \includegraphics[scale=1]{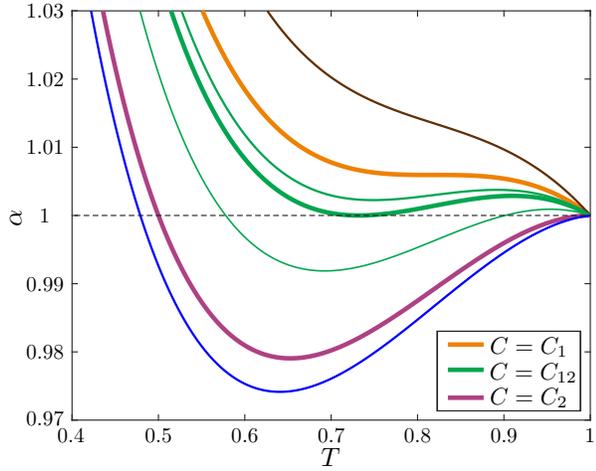} } 
    \caption{ Graph of $\tilde{\alpha}(T)$ with $a=1$, $b>0$ and various $C>0$.}
    \label{fig:Graph_TAlpha_Func}
\end{figure}

The graphs of $\talpha(T)$ for various values of $C$ with $b\leq b_c$ and $b> b_c$ are given in \cref{fig:TAlpha_Func_b_1p25,fig:TAlpha_Func_b_2}, respectively. From \cref{fig:EvolutionBifCurve_Varying_C}, we see clearly that $\tilde{S}$ is $C$-shaped (resp., $S$-shaped) as long as $C\geq C_2$ (resp., $b> b_c$ and $C_1< C < C_2$). Otherwise, it is strictly increasing. For $C_1< C< C_{12}$ (resp., $C_{12}< C< C_2$), the corresponding $\tilde{S}$ has the property that the discontinuous eradication transition is from the high prevalence state to the low prevalence state (resp., the disease free state). Hence, to further distinguish these two types of $S$-shaped curves, we shall term the corresponding $\tilde{S}$ for $C\in (C_1, C_{12})$ and $C\in (C_{12}, C_2)$ $S_r$ and $S_l$-shaped, respectively. From our summary above, we see that nonuniformity in the model is the primary reason for the occurrence of the hysteresis. To see this, in \cref{fig:BifPlane_fixed_a_or_b}, we classify the shape of $\tilde{S}$ by the uniformity measurements $a,b,C$ on the $(a,C)$ and $(b,C)$ planes with the other parameter being fixed. In \cref{fig:TAlpha_Func_b_25p5,fig:EvolutionBifCurve_Varying_a}, we consider a coinfection model of the mixed type. Specifically, we consider the parameters $a$, $b$ and $C$ satisfying $0<C<1$ and $a,b>1$. From \cref{fig:EvolutionBifCurve_Varying_C,fig:BifPlane_fixed_a_or_b,fig:TAlpha_Func_b_25p5,fig:EvolutionBifCurve_Varying_a}, we arrive at the following conclusions:

\begin{figure}[htp]
    \subfloat[ Curve $\tilde{S}$ on the $(\alpha, 1-S)$ plane with parameters $a,b$ and $C$ chosen as those in \cref{fig:TAlpha_Func_b_1p25}.\label{fig:EvoBifCurve_b_1p25}]%
             { \includegraphics[scale=1]{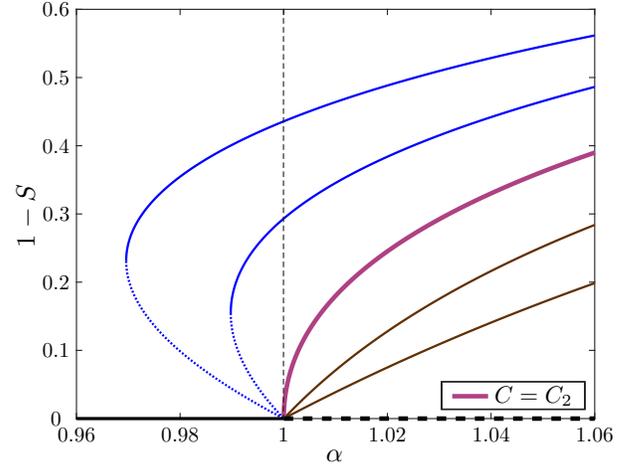} } 
    \\
    \subfloat[ Curve $\tilde{S}$ on the $(\alpha, 1-S)$ plane with parameters $a,b$ and $C$ chosen as those in \cref{fig:TAlpha_Func_b_2}.\label{fig:EvoBifCurve_b_2}]%
             { \includegraphics[scale=1]{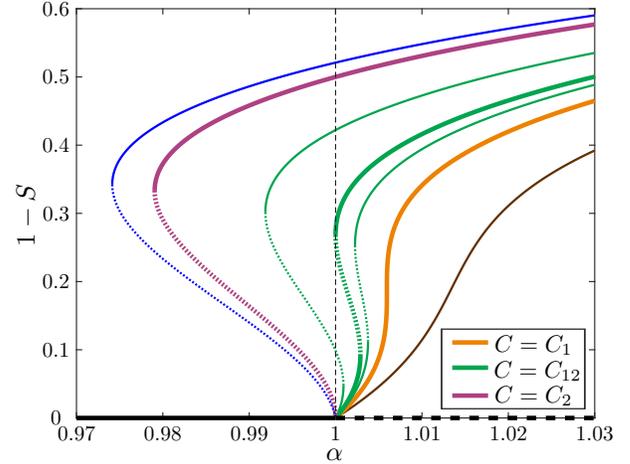} } 
    \caption{Endemic fixed point curve $\tilde{S}$ on the $(\alpha, 1- S)$-plane with fixed $a=1$, $b>0$ and varying $C>0$.}
    \label{fig:EvolutionBifCurve_Varying_C}
\end{figure}

\begin{enumerate}[label=(\Alph*),ref=(\Alph*)]
  \item As mentioned earlier, $C$, $a$, $b$ are the measurement for the degree of uniformity in the coinfection model. The assumptions that $b\leq b_c$ and $C\leq C_2$, or $b>b_c$ and $C\leq C_1$ indicate that the model under consideration is relatively uniform. Model exhibits a continuous transition at $\alpha=1$, the same threshold as the single disease spreading according to SIS model. In short, for such system, it acts as non interacting. Mathematically, the system undergoes a transcritical bifurcation at $\alpha=1$.
  \item If the relative infection factor $C$ is sufficiently large, i.e., $C> C_2$, then the model exhibits the hysteresis of type $C$ regardless of what relative coinfection factors $a$ and $b$ are. In addition to a transcritical bifurcation at $\alpha=1$, a saddle-node bifurcation occurs at $\alpha= \alpha_e$ that generates a bistable region. A hysteresis loop is found on the interval $[\alpha_e, 1]$. The case $C=C_2$ connecting continuous and discontinuous transition from the disease free state has a pitchfork bifurcation at $\alpha=1$, which is supercritical as $b \leq b_c$ while subcritical as $b> b_c$.
  \item If the difference in $b$ and $a$ becomes greater than $1/2$, i.e, $b> b_c$, then hysteresis of types $S_r$ and $S_l$ occur for $C\in (C_1, C_2)$. Two saddle-node bifurcations occur at $\alpha= \alpha_e$ and $\alpha= \alpha_o$. The hysteresis loop for both types is then found on the interval $[\alpha_e, \alpha_o]$. As mentioned in the introductory section, the occurrence of any type of hysteresis presents greater challenge for health departments. The implication of our new results is that if the relative coinfection factors differ by less than $1/2$, hysteresis of any type will never occur in an intermediate range of $C$, \textit{i.e.}, $C_1< C< C_2$. Otherwise, hysteresis of types $S_l$ and $S_r$ will take place in such range. It is worthwhile to point out that the existence of hysteresis of $S_l$ or $S_r$ types guarantees the existence of two stable endemic fixed points in some range of transmission rates while only one stable endemic fixed point is needed for hysteresis of $C$.
  \item If the devastation caused by the diseases is measured by how large the size of infected population jumps in the outbreak and by how much effort needed to eradicate the diseases, then hysteresis of type $C$ (resp., $S_r$) is the most (resp., least) devastating.
  \item For models of a mixed type, we see, via \cref{fig:TAlpha_Func_b_25p5,fig:EvolutionBifCurve_Varying_a}, that three types of hysteresis still exist provided that the model is extremely nonuniform. More precisely, with $(b,C)=(25.5, 0.082)$ by increasing the nonuniformity measurement in $a$, we see that the model goes through the following phase transition: continuous phase transition to hysteresis of $S_r$ type to continuous phase transition to hysteresis of $S_r$ type to hysteresis of $S_l$ type to hysteresis of $C$ type.
  \item The most interesting case, where $b-a>1/2$, exhibits four different regimes in the parameter space, see \cref{thm:main_results}, \cref{fig:BifPlane_a=1,fig:TAlpha_Func_b_3}. Those behaviours can be understood taking into account the microscopic features of the dynamics. In this situation, the cooperation mainly comes through state AB, which are considerably more virulent (\cref{table:Relation_p_alpha}). First, in the case that the value of $C$ is small, there are not enough AB states at the transition point in which the disease breaks out, so the transition is continuous. Moreover, enough AB states appear only when the system is already in a state with most of nodes infected. This makes the effect on cooperation small and hence we have a continuous branch. If the value of $C$ is further increased, the appearance of more AB states is stimulated. Eventually, in the case that enough AB states appear after a threshold $\alpha_o$, which is greater than the point for the outbreak of diseases, is reached, cooperation for doubly infected AB infecting other states starts to be effective, and then an abrupt transition emerges. Finally, if we keep on increasing $C$, the cooperation for secondary infection dominates the dynamics and the model matches the results reported in Ref.\ \cite{Chen2016}.
  \item In the case that the relative infection factor $a$ for susceptible class is relatively higher than that of $b$ for singly infected class, or equivalently, $b\leq a+1/2$, there exist only two different regimes in the parameter space. The reason for the disappearance of $S_r$ and $S_l$ shape is due to the fact that a higher $a$ coupled with a larger $C$ would simulate enough states AB at the transition point where the disease breaks out. Consequently, an abrupt transmission emerges.
\end{enumerate}

\begin{figure}[htp]
    \subfloat[ $a= 1$.\label{fig:BifPlane_a=1}]%
             { \includegraphics[scale=1]{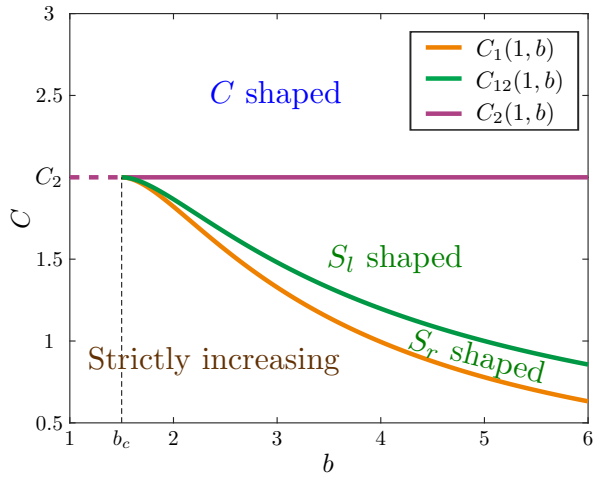} } 
    \\
    \subfloat[ $b= 1.25 < b_c$ whenever $a\geq 1$.\label{fig:BifPlane_b_1p25}]%
             { \includegraphics[scale=1]{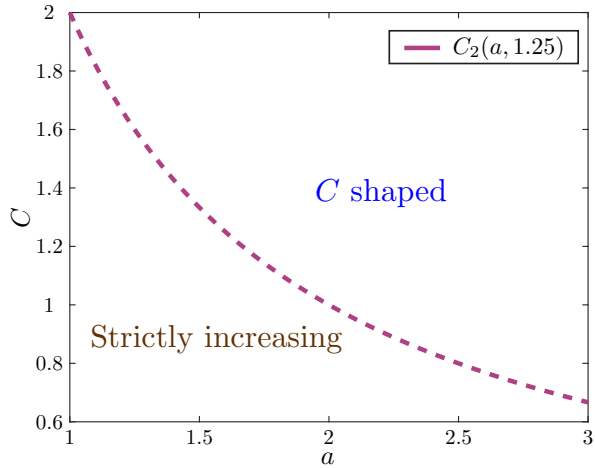} } 
    \\
    \subfloat[ $b= 3 > b_c$ whenever $a< 2.5$, and $b= 3 \leq b_c$ whenever $a\geq 2.5$.\label{fig:TAlpha_Func_b_3}]%
             { \includegraphics[scale=1]{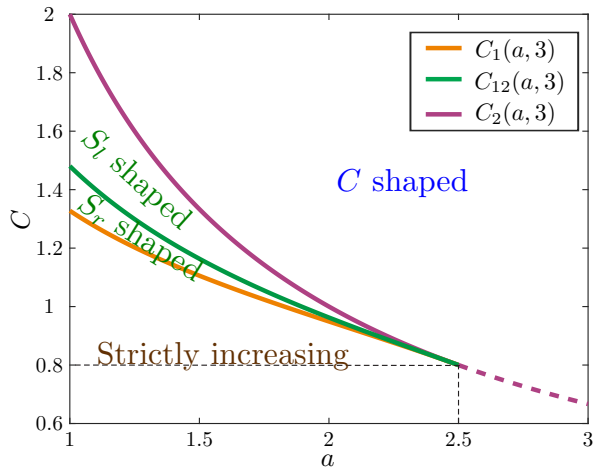} } 
    \caption{Classification of $\tilde{S}$ on the $(a,C)$ and $(b,C)$ planes with the other parameter being fixed.}
    \label{fig:BifPlane_fixed_a_or_b}
\end{figure}

\begin{figure}[htp]
  \centering
  \includegraphics[scale=1]{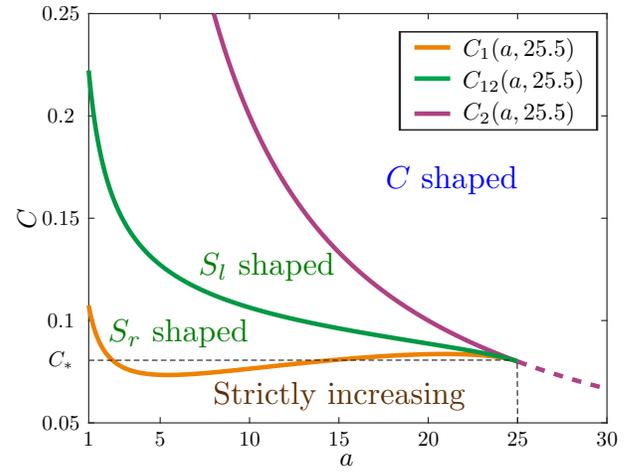} 
  \caption{Classification of $\tilde{S}$ on the $(a,C)$ plane with $b=25.5$ and $C< 1$, which corresponds to a mixed type coinfection model.}\label{fig:TAlpha_Func_b_25p5}
\end{figure}

\begin{figure}[htp]
  \centering
  \includegraphics[scale=1]{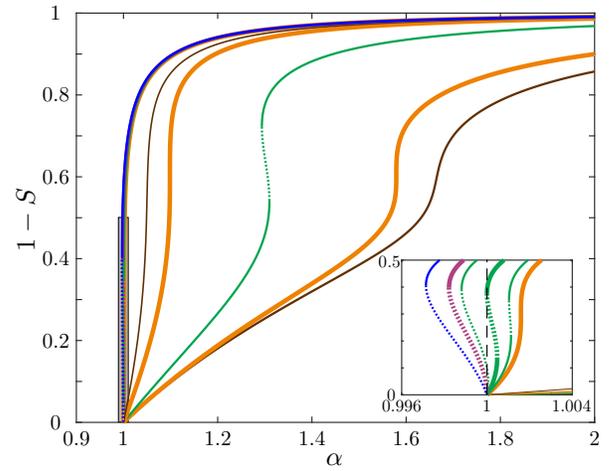} 
  \caption{$\tilde{S}$ on the $(\alpha, 1- S)$ plane with fixed $b= 25.5\, (=: b_*)$, $C= 0.082\, (=: C_*)$ and varying $a>0$ discretely from $1.5< a_1< 7< a_2< 20< a_3< 24.1< a_4< 24.33< a_5< 24.5$. Here $a_1\approx 2.1643$, $a_2\approx16.4997$, $a_3\approx24.0397$ are the three positive roots of $C_1(a, b_*)= C_*$ (see \cref{fig:TAlpha_Func_b_25p5}), and $a_4\approx 24.2080$ and $a_5\approx24.3902$ are the roots of $C_{12}(a_4, b_*)=C_*$ and $C_{2}(a_5, b_*)=C_*$, respectively. The eleven curves in the panel from the right to the left corresponds to various of $a$'s as given above, respectively. The zoom in version of the curves near $\alpha=1$ is displayed on the lower right part of the panel.}
  \label{fig:EvolutionBifCurve_Varying_a}
\end{figure}

\section{Stochastic network approach}\label{sec:STOCHASTIC_NETWORK_MODEL}

  We consider a random network where diseases spread through the links that specify who interacts with whom. If the transmission probability across a link $i$ is $\bar{p}_i$, then the probability that a node is infected through a link with one of its $m$ infected neighbors is $1- \prod\limits_{i=1}^{m} (1- \bar{p}_i)$. An infected node recovers from every infection with a uniform probability $r=0.5$. Having two different diseases, we denote by, respectively, the first type, the second type and the doubly infected type A, B and AB. We assume that susceptible nodes are infected by diseases A and B due to its interaction with an individual from class A or B with probability $p$. We also assume that a susceptible node is infected by disease A or B due to its interaction with an individual from class AB with probability $1-(1-p)^a$. We further assume that individuals in class A (resp., B) are infected with the complementary disease by individuals from classes B and AB (resp., A and AB) with probabilities $1-(1-p)^C$ and $1-(1-p)^{bC}$. The basic reaction scheme underlying the stochastic model, comparing with the well-mixed model, is given in \cref{table:Relation_p_alpha}.

  For $a=b=C=1$, the system corresponds to two noninteracting infections. The case with $C>1$, $a=b=1$, corresponding to a \emph{partially cooperative coinfection model}, has been addressed in Ref.\ \cite{Chen2016}, where cooperative coinfection exhibits \emph{hysteresis of $C$ type}.

  We simulate the stochastic dynamics in an Erdos-Renyi network of size $N=2^{14}$ with average degree $\langle k \rangle=4$. In our simulations, the primary infection probability  $p$ is varied between $p_0$ and $p_f$, both with increasing and decreasing $p$. The initial condition for the simulations at $p=p_0$ when increasing $p$ and at $p=p_f$ when decreasing $p$ is set to be the same. Specifically, we set it with a fraction $\epsilon=10^{-4}$ doubly infected that is in state AB and the rest of the nodes in susceptible state. When the stationary value is reached for every value of $p$, we measure the fraction of infected nodes $1-S$ and then we use that final state for $p$ as an initial condition for next value $p+\Delta p$. However, this approach has absorbing states for individual diseases, which are states from which you cannot leave when one of the diseases has no infected nodes. In that case, for the next value of $p$, we keep final condition for the active disease as the corresponding initial condition and set the initial condition for the disease that reached the absorbing state with a fraction $\epsilon$.

  We set the values $b=2$ and $a=1$ for our parameters, assessing how the dynamics is influenced by the changes in $C$. To avoid including in our statistics realizations in which one disease is in the absorbing state, which could bias our results towards the results of one single disease spreading, we do not include points in which this happens. First of all, as $C$ is increased, a discontinuous transition appears, where the transition point is smaller, with a higher transition gap, for higher values of $C$; see \cref{fig:resultsER}. Secondly, the jump for increasing $p$ happens at a higher value of $p$ than in the case of decreasing it; see \cref{fig:resultsER} inset. This leads to hysteresis of $C$ type, while we do not find the other two types reported for the well-mixed dynamics.

\begin{figure}[htp]
  \centering
  \includegraphics[width=0.5\textwidth]{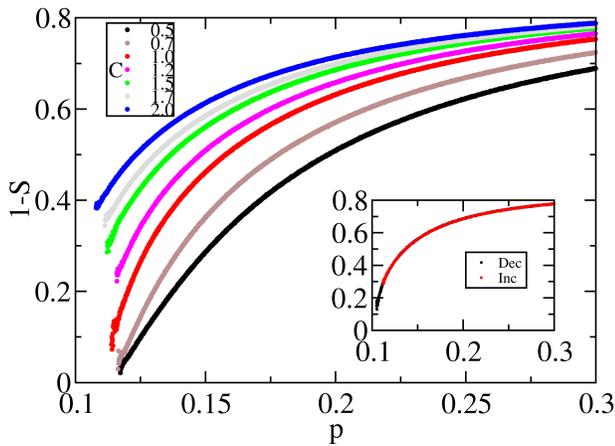} 
  \caption{Fraction of infected nodes for the fully cooperative coinfection model in an Erdos-Renyi graph with $N=2^{14}$ nodes, $\langle k \rangle=4$, averaging over 100 realizations for different values of $C$, $b=2$ and $a=1$, when $p$ is increased. Inset: for $C=1.5$, the transition point is different for increasing and decreasing $p$ (red and black points, respectively), leading to a bistability region that forms a hysteresis loop.}\label{fig:resultsER}
\end{figure}


\section{CONCLUSION}\label{sec:CONCLUSION}

In this paper, we propose a fully cooperative coinfection model, which is investigated through well-mixed and network-based approaches. For the former approach, we prove analytically that the model exhibits three types of hysteresis provided that the model is relatively nonuniform in its parameters. However, we were able to identify just one type of hysteresis in the simulations on random networks.

It is of interest for future work to study the effect of vaccination and awareness on such fully cooperative coinfection model. Moreover, finding the topologies in which $S_r$ and $S_l$ hysteresis appear will underline the microscopic mechanisms that underscore this dynamics.

\revappendix*
\setcounter{section}{0}
\renewcommand{\theproposition}{\Alph{section}.\arabic{proposition}}
\renewcommand{\theequation}{\Alph{section}\arabic{equation}}

\section{Proofs of Some Claims in \cref{sec:Well-Mixed MODEL}}
\label[app_sec]{sec:Appendix_Existence_of_bar_alpha}

In \cref{sec:Appendix_Existence_of_bar_alpha}, we provide the rigorous proof for some claims in the main content.
We begin with proving that the solutions $(S,\alpha)$ of $f(S,\alpha)=0$ with $0<S<1$
implicitly defines a function $\balpha (s)$ from to $(0,1/S]$ such that $f(S, \balpha(S))=0$.

\begin{proposition}\label{prop:exist_func_balpha}
  Fix $a,b\geq 1$ and $C>0$. Then, for any $0< S\leq 1$,
  there exists a unique $\balpha\, (= \balpha(S))$ on $(0,1/S]$ such that $f(S,\balpha(S))=0$.
\end{proposition}

\begin{proof}
  Fix $a,b\geq 1$ and $C>0$. We first consider the case $0<S<1$.
  Writing $f(S,\alpha)$ as $g_S(\alpha)$, we have that $g_S(0)= -1$ and $g_S(1/S)= a^2C(1-S)/S >0$ for $0<S<1$.
  Hence, there must exist a $\balpha$ on $(0,1/S)$ such that $g_S(\balpha)=0$.
  Suppose that there exist two roots $\alpha_1$ and $\alpha_2$ of $g_S(\alpha)$ on $(0,1/S)$.
  Then, since $g_S(0)= -1 <0$ and $g_S'(0)= -2(a-1)S \leq 0$, there exists a critical point on $[0, \alpha_1)$.
  Moreover, another critical point should exist on $(\alpha_1, \alpha_2)$. Noting $g_S(1/S) >0$ for $0<S<1$, we also
  conclude that a third critical point exists on $(\alpha_2, 1/S)$, a contradiction to the fact that
  $g_S(\alpha)$ is a cubic polynomial.
  Suppose there exist three roots on $(0,1/S)$. Upon using the fact that $g_S(0)<0$, we conclude that
  $g_S(\alpha)$ must be increasing on $(-\infty, 0)$, a contradiction to the fact that $g_S'(0)\leq 0$.
  Therefore, $g_S(\alpha)$ has a unique root on $(0,1/S)$ for $0<S<1$.
  For $S=1$, the assertion of the proposition can be verified directly.
\end{proof}

We next show that the graphs of $\balpha(S)$ and $\talpha(S)$ defined in \eqref{eq:defi_tilde_alpha}
have the same shape on the interval of $(0,1]$ in the following sense.
Note that the transformation given in \eqref{eq:defi_T} defines a bijective mapping
between $(S, \balpha(S))$ and $(T, \talpha(T))$.

\begin{proposition}\label{prop:relation_balpha_talpha}
  For any $0<t<1$, we have that
  \[
        \sgn \left( \left[ \frac{d}{dT} \talpha(T) \right]_{T=t}  \right) =
        \sgn \left( \left[ \frac{d}{dS} \balpha(S) \right]_{\textstyle S=\frac{t}{\talpha(t)}}  \right).
  \]
  Here function $\sgn(x)$ takes the sign of $x$.
\end{proposition}

\begin{proof}
  By the chain rule, we have that, for any $0<t<1$,
  \[
        \left[ \frac{d}{dS} \balpha(S) \right]_{\textstyle S=\frac{t}{\talpha(t)}} =
        \left[ \frac{ \frac{d}{dT} \talpha(T) }%
                    { \frac{d}{dT} \left[ \frac{T}{\talpha(T)} \right] }
        \right]_{T=t}.
  \]
  To complete the proof, it suffices to show that $\frac{d}{dT} \left[ \frac{T}{\talpha(T)} \right] >0$ for all $0<T<1$.
  Indeed, some direct computation yield that
  \[
    \frac{d}{dT} \left[ \frac{T}{\talpha(T)} \right] =
    \frac{N_1(T,a,b)}%
         {aCT [\talpha(T)]^2},
  \]
  where
  \[
    \begin{aligned}
      N_1(T,a,b) &= [2(2-T) T^2]a^2\\
                 &\;\;+ \{ T(1-T) [2b(1-T) +(3-T)] \} a\\
                 &\;\;+ b (2-T)(T-1)^2\\
                 &>0.
    \end{aligned}
  \]
  We have just completed the proof of the proposition.
\end{proof}

To see the shape of the graph of $\talpha(T)$,
we first compute the derivatives of $\talpha(T)$ at its two end points, $0$ and $1$.

\begin{proposition}\label{prop:diff_talpha_endpt}
  The following assertions hold.
  \begin{enumerate}[label=(\roman*),ref=(\roman*)]
        \item\label{item:talpha_endpt}
             $\lim_{T\rightarrow 0^+} \talpha(T)= \infty$ and $\talpha(1)= 1$.
        \item\label{item:dtalpha_endpt}
             $\lim_{T\rightarrow 0^+} \talpha'(T)= -\infty$ and $\talpha'(1)= 1-\frac{2}{aC}$.
        \item\label{item:ddtalpha_endpt}
             $\lim_{T\rightarrow 0^+} \talpha''(T)= \infty$ and $\talpha''(1)= -\frac{4}{a^2C} [b-a-1/2]$.
        \item\label{item:dddtalpha_endpt}
             $\talpha'''(T)< 0$ for all $0<T\leq 1$ and $a\geq 1$.
  \end{enumerate}
\end{proposition}

\begin{proof}
  Some direct computations yield that
  \begin{equation}\label{eq:formula:diff_talpha}
     \begin{aligned}
        \talpha'(T)   &=   \frac{M_1(T,a,b,C)}{ aCT^2 [(a-b)T+b]^2},\\
        \talpha''(T)  &=  2 \frac{M_2(T,a,b)  }{ aCT^3 [(a-b)T+b]^3},\\
        \talpha'''(T) &=  6 \frac{M_3(T,a,b)  }{ aCT^4 [(a-b)T+b]^4},\\
     \end{aligned}
  \end{equation}
  where
  \begin{equation}\label{eq:formula:diff_talpha_upper}
     \begin{aligned}
        M_1(T,a,b,C) &= [a(a-b)^2C]T^4+ [2ab(a-b)C]T^3\\
                     &\;\;\;+ [ab^2C-2a^2+2a-b]T^2\\
                     &\;\;\;+ 2(b-a) T -b,\\
        M_2(T,a,b)   &= [(a-b)(2a^2-2a+b)]T^3\\
                     &\;\;\;+ 3(a-b)^2T^2+3b(a-b)T+b^2,\\
        M_3(T,a,b)   &= - [(a-b)^2(2a^2-2a+b)]T^4\\
                     &\;\;\;-4(a-b)^3 T^3 -6b(a-b)^2 T^2\\
                     &\;\;\;-4b^2(a-b)T-b^3.
     \end{aligned}
  \end{equation}
  Then the assertions of parts \cref{item:talpha_endpt,item:dtalpha_endpt,item:ddtalpha_endpt}
  in the proposition can be verified directly.
  To see \cref{item:dddtalpha_endpt}, we rewrite $M_3(T,a,b)$ as the polynomial form of $b$.
  Specifically, we have that
  \begin{align*}
    M_3(T,a,b) &= c_3 b^3+ c_2 b^2+ c_1 b+ c_0\\
               &= b^2 \left[ c_0 b^{-2}+ c_1 b^{-1}+ c_2 \right]+ c_3 b^3,
  \end{align*}
  where $c_3= -(1-T)^4\, (<0)$, $c_2=  -2aT [(a-2)T^3+6T^2-6T+2]$, $c_1= a^2 T^2 [(4a-5)T^2+12T-6]$,
  $c_0= -2a^3 T^3 [(a-1)T+2]\, (<0)$.
  Hence, to conclude the assertion \cref{item:dddtalpha_endpt} in the proposition,
  we only need to show that $c_1^2- 4c_0c_2< 0$ for all $0<T\leq 1$ and $a\geq 1$. Indeed, we compute that
  $c_1^2- 4c_0c_2 = a^4 T^4 [ (8a-7) (1-T)^4- 12 (1-T)^3- 6 (1-T)^2-4(1-T)+ (1-8a) ] \leq
                    a^4 T^4 [ (8a-7)+ (1-8a) ]= -6a^4 T^4<0$.
  The proof is now complete.
\end{proof}

\begin{proposition}\label{prop:b_less_than_b_c}
  For any $a\geq 1$ and $b\leq b_c (:= a+1/2)$, the following two assertions hold.
  \begin{enumerate}[label=(\roman*),ref=(\roman*)]
        \item\label{item:case_b_less_b_c_and_C_less_C_2}
             If $C\leq C_2$, then $\talpha(T)$ is a strictly decreasing function of $T$ on $(0,1]$.
        \item\label{item:case_b_less_b_c_and_C_larger_C_2}
             If $C> C_2$, then $\talpha(T)$ has a unique minimum on $(0,1)$.
             In particular, the graph of $\talpha(T)$ is $C$-shaped.
  \end{enumerate}
  Here $C_2= \frac{2}{a}$ is define as in \eqref{eq:defi_C_2}.
\end{proposition}

\begin{proof}
  Since $b\leq b_c$, we have that $\talpha''(1)\geq 0$ by \cref{prop:diff_talpha_endpt}\cref{item:ddtalpha_endpt}.
  Consequently, $\talpha''(T)> 0$ for all $0<T<1$ by \cref{prop:diff_talpha_endpt}\cref{item:dddtalpha_endpt},
  which implies that $\talpha(T)$ is concave up on $(0,1]$.
  In addition, $\talpha(T)$ is strictly decreasing on $(0,1]$
  if $\talpha'(1) \leq 0$,
  while first decreasing and then increasing on $(0,1]$ if $\talpha'(1) > 0$.
  Above cases correspond to $C\leq C_2$ or $C> C_2$, respectively, by \cref{prop:diff_talpha_endpt}\cref{item:dtalpha_endpt}.
  Hence, the proof is complete.
\end{proof}

For $a\geq 1$ and $b>b_c$, the graph of $\talpha(T)$ is more complicated since
the concavity of $\talpha(T)$ changes from the upward to the downward on $(0,1)$
by \cref{prop:diff_talpha_endpt}\cref{item:ddtalpha_endpt,item:dddtalpha_endpt}.
More precisely, $\talpha(T)$ has an inflection point $\gamma$ on $(0,1)$ such that
\begin{equation}\label{eq:inflection_pt_talpha}
  \talpha''(T)
  \left\{
    \begin{array}{cl}
      >0 & \text{if } T\in (0, \gamma), \\
      =0 & \text{if } T= \gamma, \\
      <0 & \text{if } T\in (\gamma,1]. \\
    \end{array}
  \right.
\end{equation}
To classify its shape, we first find $C$ so that the corresponding $\talpha(T)$ has the some particular
properties.

\begin{proposition}\label{prop:defi_C1_C_12}
  For any $a\geq 1$ and $b> b_c\, (:= a+1/2)$, define $\gamma$ as in \eqref{eq:inflection_pt_talpha}.
  Then the following three assertions hold.
    \begin{enumerate}[label=(\roman*),ref=(\roman*)]
        \item\label{item:existence_C1}
             There exists a unique positive $C_1$, defined as in \eqref{eq:defi_C_1}, such that
             \begin{equation}\label{eq:property_C1}
                \talpha'(\gamma)= \talpha''(\gamma)=0.
             \end{equation}
        \item\label{item:existence_C12}
             There exists a unique positive $C_{12}$, defined as in \eqref{eq:defi_C_12}, such that the equation
             \begin{equation}\label{eq:property_C12}
                \talpha(T)=1\, (=\talpha(1))
             \end{equation}
             has a unique solution on $(0,1)$ when $C=C_{12}$.
             Moreover, \eqref{eq:property_C12} has no solution on $(0,1)$ when $C< C_{12}$.
        \item\label{item:relation_C1_C12_C2}
             It holds that $C_1< C_{12}< C_2$.
             Here $C_2= \frac{2}{a}$ is define as in \eqref{eq:defi_C_2}.
    \end{enumerate}
\end{proposition}

\begin{proof}
  Fix $a\geq 1$ and $b> b_c (:= a+1/2)$. We first show \cref{prop:defi_C1_C_12}\cref{item:existence_C1}.
  Since $\gamma$ is the inflection point of $\talpha$, it holds that $\talpha''(\gamma)=0$.
  We have, via \eqref{eq:formula:diff_talpha_upper}, that $M_2(\gamma,a,b)=0$.
  Moreover, since function $M_2$ is independent of $C$,
  we have that $\gamma$ is also independent of $C$.
  On the other hand, note that
  \begin{subequations}\label{eq:diff_talpha_M_1}
    \begin{equation}\label{eq:sign_diff_talpha_M_1}
        \sgn(\talpha'(T))= \sgn(M_1(T,a,b,C))
    \end{equation}
  by \eqref{eq:formula:diff_talpha_upper}, and
  $M_1(T,a,b,C)$ can be rewritten as
  \begin{equation}\label{eq:M_1}
    \begin{aligned}
      M_1(T,a,b,C)&= [aT^2 (aT-bT+b)^2]C\\
                  &\;\;\;- [ (2a^2-2a+b) T^2+ 2(b-a)T+b ]\\
                  &=: \alpha_1(T)C - \alpha_0(T).
    \end{aligned}
  \end{equation}
  \end{subequations}
  Clearly, $\alpha_1(T)$ and $\alpha_0(T)$ are positive for all $0<T<1$.
  By choosing $C_1= \frac{\alpha_0(\gamma)}{\alpha_1(\gamma)}$, we have that
  \begin{equation}\label{eq:diff_sign_talpha_gamma}
        \talpha'(\gamma)
            \left\{
                \begin{array}{cl}
                    < 0 & \text{if } C< C_1, \\
                    = 0 & \text{if } C= C_1, \\
                    > 0 & \text{if } C> C_1. \\
                \end{array}
            \right.
  \end{equation}
  In fact, we can obtain the explicit form of $C_1$, as given in \eqref{eq:defi_C_1} by using the Maple.

  We next show \cref{prop:defi_C1_C_12}\cref{item:existence_C12}. By solving $\talpha(T)=1$ directly, we have that
  $T=1$ or
  \[
    T= \frac{(abC-2a+1) \pm \sqrt{L(a,b,C)}}{2a(b-a)C}=: \hat{T},
  \]
  where
  \[
    L(a,b,C)= (a^2b^2) C^2+ (-4a^2b+4a^2-2ab)C+ (2a-1)^2.
  \]
  Hence, for equation $\talpha(T)=1$ to have exactly one solution on $\mathbb{R}$, we must have $L(a,b,C)=0$
  or, equivalently,
  \[
        C= \frac{ (2ab-2a+b)\pm 2 \sqrt{a (2b-1) (b-a)}}{ab^2} =: C_{\pm}.
  \]
  We are then able to derive the following results.
  \begin{enumerate}[label=(\alph*),ref=(\alph*)]
    \item For $C=C_{-}$, we have that $\hat{T}= \frac{(abC_{-}-2a+1)}{2a(b-a)C_{-}} <0$
          by some tedious computation.
          It implies that $\talpha(T)$ is either greater or less than $1$ on the entire interval $(0,1)$.
          In fact, $\talpha(T)>1$ on $(0,1)$ since $\lim_{T\rightarrow 0^+} \talpha(T)= \infty$
          by \cref{prop:diff_talpha_endpt}\cref{item:talpha_endpt}.
    \item For $C< C_{-}$, we have, via direct computation, that
          \begin{equation}\label{eq:diff_talpha_C}
                \frac{\partial}{\partial C} \talpha(T) = - \frac{(1-T)[2aT+(1-T)]}{aC^2T[aT+b(1-T)]} < 0
          \end{equation}
          whenever $0<T<1$.
          This implies that $\left[ \talpha(T)  \right]_{C} > \left[ \talpha(T) \right]_{C=C_-} > 1$ for $C< C_-$.
    \item For $C\in (C_-, C_+)$, since $L(a,b,C)<0$, equation $\talpha(T)=1$ has no solution on $\mathbb{R}$ and hence
          $[\talpha(T)]_{C}>1$ on $(0,1)$.
    \item For $C=C_{+}$, we have that $\hat{T}= \frac{(abC_{+}-2a+1)}{2a(b-a)C_{+}} \in (0,1)$
          by some tedious computation.
          Moreover, since $[\talpha(T)]_{C}>1$ on $(0,1)$ for $C< C_{+}$ as claimed above,
          we have that  $[\talpha(T)]_{C_{+}} \geq 1$ on $(0,1)$ and the equality holds if and only if $T=\hat{T}$.
  \end{enumerate}
  Hence, the proof of \cref{prop:defi_C1_C_12}\cref{item:existence_C12} is complete by choosing $C_{12}= C_{+}$.

  Finally, we show \cref{prop:defi_C1_C_12}\cref{item:relation_C1_C12_C2}.
  As in the proof of part \cref{item:existence_C12}, we have that
  $\talpha(T) > 1$ on $(0,1)$ whenever $C<C_{12}$; while
  $\talpha(T) \geq 1$ on $(0,1)$ and there exists $T_*$ on $(0,1)$ such that $\talpha(T_*)= 1$ whenever $C=C_{12}$;
  $\talpha(T) < 1$ for some point on $(0,1)$ whenever $C>C_{12}$ by \eqref{eq:property_C12} and \eqref{eq:diff_talpha_C}.
  Then, by \eqref{eq:inflection_pt_talpha} and \eqref{eq:property_C1}, we have that, when $C=C_1$,
  $\talpha(T)$ is strictly decreasing on $(0,1)$ and hence $\talpha(T)> \talpha(1)\, (=1)$ on $(0,1)$.
  So $C_1< C_{12}$.
  On the other hand, since $\talpha'(1)=0$ by \cref{prop:diff_talpha_endpt}\cref{item:dtalpha_endpt}
  and $\talpha''(1)<0$ by \eqref{eq:inflection_pt_talpha}, we have that $\talpha'(T)<0$ and hence
  $\talpha(T)< \talpha(1)\, (=1)$ whenever $T$ is sufficiently close
  to $1^-$. So $C_2> C_{12}$.

  We have just completed the proof of \cref{prop:defi_C1_C_12}.
\end{proof}

\begin{proposition}\label{prop:talpha_property_various_C}
  For any $a\geq 1$ and $b> b_c (:= a+1/2)$, the following two assertions hold.
      \begin{enumerate}[label=(\roman*),ref=(\roman*)]
        \item\label{item:case_b_larger_b_c_and_C_less_C_1}
             If $C<C_1$, then $\talpha(T)$ is strictly decreasing.
        \item\label{item:case_b_larger_b_c_and_C_in_C_1_C_12}
            If $C\in (C_1, C_{12})$, then $\talpha(T)$ has a unique local minimum and
            a unique local maximum on $(0,1)$. The graph of such $\talpha(T)$ is to be called an $S$-shaped.
            Moreover, we have that $\talpha(T)> 1$ for all $0<T<1$.
        \item\label{item:case_b_larger_b_c_and_C_in_C_12_C_2}
            If $C\in (C_{12}, C_2)$, then the graph of $\talpha(T)$ is also $S$-shaped with $\talpha(T)<1$.
            for some $T$ on $(0,1)$.
        \item\label{item:case_b_larger_b_c_and_C_larger_C_2}
            If $C> C_2$, then $\talpha(T)$ has a unique critical point. The graph of such $\talpha(T)$ is to be termed
            a $C$-shaped.
    \end{enumerate}
\end{proposition}

\begin{proof}
  We first prove \cref{prop:talpha_property_various_C}\cref{item:case_b_larger_b_c_and_C_less_C_1}.
  By \eqref{eq:inflection_pt_talpha} and \eqref{eq:property_C1}, we have that, when $C=C_1$,
  $\talpha'(T) \leq 0$ on $(0,1)$. Then, by \eqref{eq:diff_talpha_M_1}, we have that $\talpha'(T) < 0$ on $(0,1)$
  if $C<C_1$.
  So the assertion of part \cref{item:case_b_larger_b_c_and_C_less_C_1} holds.

  We next prove
  \cref{prop:talpha_property_various_C}%
  \cref{{item:case_b_larger_b_c_and_C_in_C_1_C_12},item:case_b_larger_b_c_and_C_in_C_12_C_2}.
  For $C\in (C_1, C_2)$, we have that $\talpha'(\gamma)>0$ by \eqref{eq:diff_sign_talpha_gamma} and $\talpha'(1)<0$
  by \cref{prop:diff_talpha_endpt}\cref{item:dtalpha_endpt}. Then by \eqref{eq:inflection_pt_talpha},
  we conclude that $\talpha(T)$ has a unique local minimum and a unique local maximum on $(0,1)$.
  The remaining parts of the proof have been claimed in the proof of \cref{prop:defi_C1_C_12}\cref{item:relation_C1_C12_C2}.
  So the assertions of part \cref{{item:case_b_larger_b_c_and_C_in_C_1_C_12},item:case_b_larger_b_c_and_C_in_C_12_C_2} hold.

  Finally, we prove \cref{prop:talpha_property_various_C}\cref{item:case_b_larger_b_c_and_C_larger_C_2}.
  In fact, the assertion holds by \eqref{eq:inflection_pt_talpha} and
  since $\talpha'(\gamma)>0$ by \eqref{eq:diff_sign_talpha_gamma} and $\talpha'(1)>0$
  by \cref{prop:diff_talpha_endpt}\cref{item:dtalpha_endpt}.

  We have just completed the proof of \cref{prop:talpha_property_various_C}.
\end{proof}

\begin{acknowledgments}
    J.P.R. acknowledges support from Ministry of Science and Technology (Taiwan) and the FPU program of MECD (Spain), and thanks V\'ictor M. Egu\'iluz and Fakhteh Ghanbarnejad for comments and discussion. Y.H.L. and J.J. acknowledge support from the Ministry of Science and Technology of Taiwan under grant No. Most 105-2115-M-009-002-MY2.
\end{acknowledgments}


%

\end{document}